\documentclass{article}

\usepackage{research_notes}

\title{Comparing computational entropies below majority
\\ (or: When is the dense model theorem false?)}

\author{
\begin{tabular}[t]{c@{\extracolsep{6em}}c} 
	Russell Impagliazzo\thanks{Research supported by NSF award 1909634 and a Simons Investigator award}
	& Sam McGuire\thanks{Research supported by NSF award 1909634 and a Simons Investigator award}\\
UCSE CSE & UCSD CSE \\
\href{mailto:russell@cs.ucsd.edu}{russell@cs.ucsd.edu} & \href{mailto:shmcguir@eng.ucsd.edu}{shmcguir@eng.ucsd.edu}
\end{tabular}
}

\begin{document}

\maketitle

\begin{abstract}
	Computational pseudorandomness studies the extent to which a random variable $\bf{Z}$ looks
	like the uniform distribution according to a class of tests $\cF$. 
	Computational entropy generalizes computational pseudorandomness
	by studying the extent which a random variable looks like a \textit{high entropy} distribution. There
	are different formal definitions of computational entropy with different advantages for different applications. 
	Because of this, it is of interest to understand when these definitions are equivalent.

	We consider three notions of computational entropy which are known to be equivalent 
	when the test class $\cF$ is closed under taking majorities. This
	equivalence constitutes (essentially) the so-called \textit{dense model theorem} of Green and Tao 
	(and later made explicit by Tao-Zeigler, Reingold et al., and Gowers). 
	The dense model theorem plays a key 
	role in Green and Tao's proof that the 
	primes contain arbitrarily long arithmetic progressions and has
	since been connected to a surprisingly wide range of topics in mathematics and computer science,
	including cryptography, computational complexity, combinatorics and machine learning. 
	We show that, in different situations 
	where $\cF$ is \textit{not} closed under majority, this equivalence fails.
	This in turn provides examples where the dense model theorem 
	is \textit{false}.
\end{abstract}

\section{Introduction}

Computational pseudorandomness is a central topic in theoretical computer science. In this scenario, one has a class $\cF$ of
boolean functions $f: \{0, 1\}^n \to \{0, 1\}$ (which we'll refer to as \emph{tests}) and random variable $\bf{Z}$ over $\{0, 1\}^n$. We say that $\bf{Z}$ is $\epsilon$-pseudorandom with respect to $\cF$) if $\max_{f \in \cF} |\E[f(\bf{Z})] - \E[f(\bf{U})]| \leq \eps$ where $\bf{U}$ is the uniform distribution over $\{0, 1\}^n$ and $\eps > 0$ is small. In this case, we think of $\bf{Z}$ as `behaving like the uniform distribution' according to tests in $\cF$.
In general, say that two random variables $\bf{X}$, $\bf{Y}$ $\eps$-indistinguishable by
$\cF$ if $\max_{f \in \cF} |\E[f(\bf{X})] - \E[f(\bf{Y})]|$ (and so $\eps$-pseudorandom distributions are exactly those which are $\eps$-indistinguishable from $\bf{U}$). Constructing explicit $\bf{Z}$'s which behave like the uniform distribution according to different test classes is among the central goals of complexity theory, with sufficiently strong constructions leading to, for example, derandomization of $\BPP$. One way in which the theory of pseudo-randomness is rich is that there are multiple equivalent formulations of pseudo-randomness, such as Yao's next bit test (\cite{yao1982theory}).

The various notions of pseudo-entropy and pseudo-density generalize pseudo-randomness to formalize how much randomness a distribution looks like it has as far as this class of tests can
perceive.  Many of these notions were first introduced as stepping stones towards pseudo-randomness, giving properties of sub-routines within constructions of pseudo-random generators.
However, measuring seeming randomness quantitatively is important in many other contexts, so these notions have found wider application.  For example, in mathematical subjects such as combinatorics and number theory, there is a general phenomenon of ``structure vs. randomness'', where a deterministically defined object such as a graph or set of integers can be decomposed into
a structured part and a random part.  Pseudo-entropy quantifies how much randomness the ``random part'' has.  Notions of pseudo-density were used in this context by Green, Tao, and Ziegler \cite{green2008primes, tao2008primes} to
show that the primes contain arbitrarily long arithmetic progressions.  We can also use pseudo-entropy notions to characterize the amount of seeming randomness remains n a cryptographic key after it has been compromised with a side-channel attack.  A data set used in a machine learning algorithm might not have much randomness in itself, and might not be completely random looking, but is hopefully representative of the much larger set of inputs that the results of the algorithm will be applied to,
so we can use notions of pseudo-entropy to say when such algorithms will generalize.   There are many possible definitions of this intuitive idea, and as with pseudo-randomness, the power of pseudo-entropy is that many of these notions have been related or proven equivalent.  

In particular, the dense model theorem provides such a basic equivalence.  Here, the
intuitive concept we are trying to capture is the density (or relative min-entropy) of the target
distribution within a larger distribution, what fraction of the larger distribution is within the
target.
We say that $\bf{Z}$ is $\delta$-dense if
$\E[\mu(x)] = 2^{-n} \sum_x \mu(x) \geq \delta$ where 
$\mu: \{0, 1\}^n \to [0, 1]$
is density function defining $\bf{Z}$ 
(in the sense that $\Pr[\bf{Z} = z] = \mu(z)/(2^n\E[\mu(x)])$).
One application of indistinguishability from a dense distribution is as a stepping stone
to pseudorandomness:   if $\bf{Z}$ is indistinguishable from a distribution $\bf{M}$ with density
$\delta$ within the uniform distribution, then applying a randomness extractor with min-entropy
rate $n - \log(1/\delta)$ to $\bf{Z}$ is a pseudorandom distribution. 
A more sophisticated application comes from additive number theory. It is not hard to show that a \textit{random} subset of $[N] = \{1, 2, ..., N\}$ (including each element with probability $1/2$, say) contains many arithmetic progressions (which are sets of the form $\{a, a + b, a + 2b, a + 3b, ...\}$).
\Szemeredi \cite{szemeredi1969sets} showed that, in fact, 
sufficiently \textit{dense} subsets of the integers also contain 
such arithmetic progressions: specifically, that for any $k$, the size of the largest subsets of  $[N]$
which \textit{doesn't} contain an arithmetic progression grows like $o(N)$. 

So we would like some technology to reason about random variables $\bf{Z}$ which `behave like dense distributions'. It turns out, however, that formalizing what it means for $\bf{Z}$ to `behave like a dense distribution' is subtle. Here are three perfectly legitimate candidates: 

\begin{description}
	\item [Candidate 1:] \ul{$\bf{Z}$ behaves like a $\delta$-dense distribution if it behaves like something that's $\delta$-dense.} Formally, this means that $\bf{Z}$ is $\eps$-indistinguishable from some $\delta$-dense distribution. In this case, we say that $\bf{Z}$ has a \textit{$\delta$-dense $\eps$-model}.
	\item [Candidate 2:]\ul{$\bf{Z}$ behaves $\delta$-dense if it's $\delta$-dense inside of something that behaves like the uniform distribution}. Formally this means there's an $\eps$-pseudorandom distribution $\bf{X}$ in which $\bf{Z}$ is $\delta$-dense. In this case,
		we say that $\bf{Z}$ is \textit{$\delta$-dense in an $\eps$-pseudorandom set}.
	\item [Candidate 3:] \ul{$\bf{Z}$ behaves $\delta$-dense if it appears to 
		be the case that conditioning on $\bf{Z}$ 
	increases the size of any set by at most (roughly) a $1/\delta$-factor.} 
	This is an operational definition: conditioning
	on a (truly) dense set increases the set by at most a $1/\delta$-fraction, so we should expect the same behavior
from things that behave like a dense set.
Formally, this means that $\delta \E[f(\bf{Z})] \leq \E[f(\bf{U})] + \eps$
for any $f$ in our test class $\cF$.
In this case, we say that $\bf{Z}$ has \textit{$(\eps, \delta)$-pseudodensity.} 
\end{description}

Precisely which definition you pick will depend on what you know about $\bf{Z}$ and in what sense
you would like it to behave like a $\delta$-dense distribution. 
Indeed, each of these definitions have appeared in different applications
(\cite{haastad1999pseudorandom}, \cite{green2008primes}, \cite{doron2019nearly}, respectively), 
so there are scenarios where each of these types of behavior is desired. 
In general, the first candidate is the strongest (and,
arguably, the most natural), but it is sometimes hard to 
establish that a distribution has the property. 
The following claim gives some simple relationships between
the definitions:

\begin{claim}
	\label{claim:intro}
	For any $\cF$, the following hold:
	\begin{enumerate}
		\item If $\bf{Z}$ has a $\delta$-dense $\eps$-model, then $\bf{Z}$ is $\delta$-dense in a $\eps$-pseudorandom set.
		\item If $\bf{Z}$ is $\delta$-dense in an $\eps$-pseudorandom set, then $\bf{Z}$ has $(\epsilon, \delta)$-pseudodensity.
	\end{enumerate}
\end{claim}

\begin{proof}[Proof sketch]
	\begin{enumerate}
		\item Let $\bf{M}$ be the $\delta$-dense $\eps$-model for $\bf{Z}$.
		Note that $\bf{U} = \delta \bf{M} + (1 - \delta) \ol{\bf{M}}$.
		So $\bf{U}' = \delta \bf{Z} + (1 - \delta) \ol{\bf{M}}$ is $\eps$-pseudorandom
		and $\bf{Z}$ is $\delta$-dense within it.

		\item Suppose $\bf{Z}$ is $\delta$-dense in $\bf{Z}'$ which $\eps$-pseudorandom for $\cF$.
		Then for any $f \in \cF$, $\delta \E[f(\bf{Z})] \leq \E[f(\bf{Z}')] \leq 
		\E[f(\bf{U}] + \eps$.
	\end{enumerate}
\end{proof}

\noindent The marvelous quality of these three candidates in particular is that,
for many natural $\cF$, all of them are \textit{equivalent}, and so establishing
even $(\epsilon', \delta)$-pseudodensity is enough to guarantee the existence of a $\delta$-dense
$\eps$-model.

This equivalence holds for $\cF$ which are \textit{closed under majority},
meaning for any $k$ (which we can think of as $k = O(1)$ for now), if $f_1, ..., f_k \in \cF$
then $\MAJ_k(f_1, ..., f_k) \in \cF$, where $\MAJ: \{0, 1\}^n \to \{0, 1\}$ is $1$ if at least half of its input bits
are $1$.
In fact, it holds for more general $\cF$ if we allow the distinguishing parameter ($\eps'$ in $(\eps', \delta)$-pseudodensity) to be exponentially small (as in the original formulation, which we'll dicuss later on).
In this case, the subtelty in defining what it means to behave like a dense set vanishes.
These equivalences constitute (essentially) what is known as the \textit{dense model theorem}, originating in the work
of Green-Tao \cite{green2008primes} and Tao-Zeigler \cite{tao2008primes}, and independently in Barak et al. \cite{barak2003computational} (though in different guises). 
This result has been fruitfully applied
in many seemingly unrelated areas of mathematics and computer science: additive number theory 
\cite{green2008primes, tao2008primes} where $\cF$ encodes additive information about subsets
of $\{1, ..., N\}$ (or possibly a more general group),
graph theory \cite{trevisan2009regularity, reingold2008dense}
where $\cF$ encodes cuts in a fixed graph,
circuit complexity \cite{trevisan2009regularity},
Fourier analysis \cite{impagliazzo2020connections}, 
machine learning \cite{impagliazzo2020connections}
and leakage-resilient cryptography \cite{dziembowski2008leakage}. 
The ubiquity of the dense model theorem motivates a simple question: 
are there natural scenarios in which the dense model theorem is \textit{false}?

We show that the answer to this question is \textit{yes}. In particular, we show that for either
implication from \Cref{claim:intro} there is a class $\cF$ and a random variable
$\bf{Z}$ so that converse fails to hold. From the computational entropy perspective, we
show that the three computational entropies we've discussed are inequivalent for certain test classes $\cF$.
Necessarily (with $\eps'$ not exponentially small)
these classes are \textit{not} closed under majority and so we will need to look `below' majority
in order to find our counterexamples. 

\subsection{The dense model theorem}

We turn to discuss the dense model theorem in some more detail to better
contextualize our work. Restricting our attention to random
variable over $\{0, 1\}^n$,  
the dense model theorem states the following:

\begin{thm}[Dense model theorem]
	\label{thm:dmt1}
	Let $\cF$ be a class of tests $f: \{0, 1\}^n \to \{0, 1\}$
	and $\bf{Z}$ a random variable over $\{0, 1\}^n$
	with $(\eps \delta, \delta)$-pseudodensity with respect to $\MAJ_{k} \circ \cF$
	for $k = O(\log(1/\delta)/\eps^2)$. Then $\bf{Z}$ has a $\delta$-dense
	$\eps$-model with respect to $\cF$.
\end{thm}

We will generally also consider a parameter $\eps'$, which in this
case is $\eps \delta$, the additive error in pseudodensity.
To get an intuition for what this is saying, let's conisder a setting where it's
false but for trivial reasons. As a simple example given in \cite{zhang2011query}, 
pick a set $\bf{Z}$ some set as a $(1-\eps)$ fraction of 
another set $\bf{S}$ of size $\delta 2^n$. Then $\bf{Z}$ doesn't have a $\delta$-dense $\eps$-model (i.e. $\bf{S}$) with
respect to $\bf{Z}$'s indicator function, which we'll call $f$. 
On the other hand, the distribution $\bf{W}$ obtained by sampling $\bf{Z}$ 
with probability $\delta$ and sampling from $\bf{S}$'s complement 
with probability $1 - \delta$ is at most $\eps \delta$-distinguishable
from $\bf{S}$ for any function, since $\eps \delta$ is 
simply the measure of the difference
between $\bf{S}$ and $\bf{Z}$. In particular $\bf{Z}$ is $\delta$-dense
in the $\eps \delta$-pseudorandom $\bf{W}$ (which implies, via \Cref{claim:intro},
that it is $(\eps\delta, \delta)$-pseudodense). This means that
the \Cref{thm:dmt1} is tight for the dependence on $\eps' = \eps \delta$,
in that it becomes false for $\Omega(\eps \delta)$. In many instances,
we think of $\eps = 1/\poly(n)$, $\delta$ constant (or perhaps with mild
dependences on $n$) and $\eps' = \delta\eps$.

Originally, the dense model theorem was proved with a different (and stronger)
assumption; namely, that $\bf{Z}$ is dense in a pseudorandom set.
Green and Tao, in proving that the primes contain arbitrarily
long arithmetic progressions, used it to the following effect:
if $\bf{Z}$ are the prime numbers up to $n$,
then its density is known to behave like $\Theta(1/\log n)$. On the other hand, \Szemeredi  
\cite{szemeredi1969sets}
showed that sufficiently dense subsets of $\Z$ contain arbitrarily long arithmetic progressions.
The best bounds for \Szemeredi's theorem require density $\omega(1/\log\log n))$, which is much 
larger than the primes (see \cite{gowers2001new} and the recent \cite{bloom2020breaking} for more
on the rich history on this and related problems). 
Not all is lost, however: the only property of dense sets that we're interested in
is that they contain arithemtic progressions. So Green and Tao construct a class $\cF$
of tests which can `detect' arithmetic progressions and 
under which the primes are dense inside of a $\cF'$-pseudorandom set (more on $\cF'$ later).
By applying the dense model theorem, we conclude that the primes `look like' a dense set (themselves having long arithemtic progressions) with respect to the class $\cF$. As $\cF$
detects arithmetic progressions, it must be the case that the primes possess them. Of course, many details need to be filled in, but we hope this example shows the reader the `spirit' of the dense model theorem.

A primary source of interest in the dense model theorem is in the connections it shares with
seemingly unrelated branches of mathematics and computer science. The original application
was in additive number theory, but it was independently discovered and proved in the context 
of cryptography (\cite{barak2003computational, dziembowski2008leakage}). RTTV \cite{reingold2008dense} and Gowers \cite{gowers2010decompositions} observed
proofs of the dense model theorem which use linear programming duality, which is in turn
related to Nisan's proof of the hardcore lemma from circuit complexity \cite{impagliazzo1995hard}. In fact,
Impagliazzo \cite{impagliazzo2020connections} shows in unpublished work that optimal-density versions of the
hardcore lemma due to Holenstein \cite{holenstein2005key} actually \textit{imply} the dense model theorem.
Klivans and Servedio \cite{klivans2003boosting} famously observed the relationship betweeen 
the hardcore lemma and \textit{boosting}, a fundamental technique
for aggregating weak learners in machine learning \cite{freund1999short}. Together with the result of Impagliazzo,
this connection means that dense model theorems can be proved by a particular type of boosting algorithm.
A boosting argument for the existence of dense models also gives us \textit{constructive} 
versions of the dense model theorem,
which are needed for algorithmic applications.
Zhang \cite{zhang2011query} (without using Impagliazzo's reduction from the dense model theorem
to the hardcore lemma) used the boosting algorithm of \cite{barak2009uniform} directly
to prove the dense model theorem with optimal query complexity ($k$). 

In addition to its connections to complexity, machine learning, additive number theory and cryptography,
the dense model theorem (and ideas which developed from the dense model theorem, 
chiefly the approximation theorem of \cite{trevisan2009regularity}), 
have been used to understand the weak graph regularity lemma of Frieze and Kannan \cite{impagliazzo2020connections},
notions of computational 
differential privacy \cite{mironov2009computational} and even generalization in generative adversarial 
networks (GANs) \cite{arora2017generalization}. 
We now turn to discussing the complexity-theoretic aspects of the dense model theorem,
 specifically regarding our question of whether the $\MAJ_k$ from the statement is optimal.

As alluded to earlier, Green and Tao actually worked in a setting where $\cF'$ doesn't 'need
to compute majorities but where $\eps \delta$ (that is, the distinguishing parameter
in the pseudodensity assumption in the statement of \Cref{thm:dmt1}) needs to
be replaced by some $\eps' = \exp(-\poly(1/\eps, 1/\delta))$ 
(with $k = \poly(1/\delta, 1/\eps)$ experiencing a small increase).
We state this result, as proved in Tao and Zeigler \cite{tao2008primes} and stated this
way in RTTV \cite{reingold2008dense}, for comparison. For a test class $\cF$,
let $\prod_k \cF$ be the set of tests of the form $\prod_{i \in [k]} f_i$
for $f_i \in \cF$. 

\begin{thm}[Computationally simple dense-model theorem, strong assumption]
	\label{thm:dmt2}
	Let $\cF$ be a class of tests $f: \{0, 1\}^n \to [0, 1]$
	and $\bf{Z}$ a random variable over $\{0, 1\}^n$
	which is $\delta$-dense in a set $\eps'$-pseudorandom for $\prod_k \cF$
	with $k = \poly(1/\delta, 1/\eps)$ and $\eps' = \exp(-1/\delta, 1/\eps)$. 
	Then $\bf{Z}$ has a $\delta$-dense
	$\eps$-model with respect to $\cF$.
\end{thm}

\noindent RTTV \cite{reingold2008dense} observe that this proof can be adapted to work for $\eps'$ have polynomial dependence
on $\eps, \delta$ by restricting to the case of boolean-valued tests. Doing so,
however, makes $\cF'$ much more complicated (essentially requiring circuits of size exponential in $k$).
In \Cref{thm:dmt1}, we can obtain the best of both worlds: $\eps'$ has polynomial dependence 
on $\eps, \delta$ and the complexity
blow-up is rather small. However, in this more picturesque circumstance, 
we need to be able to compute majorities.
Is such a tradeoff necessary? Our results suggest that the answer is yes.
\Cref{thm:pdtodm} (stated in the following section) tells us that if
the dense model theorem is true for $\cF$, 
then there's a small, constant-depth circuit with $\cF$-oracle gates approximating
majority on $O(1/\eps^2)$ bits. 

Another important aspect of the dense model theorem is how the different assumptions are related.
As mentioned, the original assumption was that $\bf{Z}$ is $\delta$-dense in an $\eps$-pseudorandom
set, but the proof can be extended to the case where $\bf{Z}$ is $(\eps, \delta)$-pseudodense.
\Cref{claim:intro} showed that the former assumption implies that latter assumption.
When the dense model theorem is true, the latter also implies the former: simply
apply the dense model theorem to $\bf{Z}$ which is $(\eps, \delta)$-dense to obtain
a $\delta$-dense $\eps$-model. Then, by the first part of \Cref{claim:intro}, we're done.

First, we give examples of situations where these two notions are distinct.
For example, we show in \Cref{thm:pdtodsops} and \Cref{thm:degreelb} that 
they are inequivalent when $\cF$ is constant-depth polynomial size circuits
or when $\cF$ is a low-degree polynomial over a finite field. 
Note that a separation between pseudodensity and being dense in a pseudorandom set also implies a separation between
pseudodensity and having a dense model, as being dense in a pseudorandom set
is a necessary condition for having a dense model.

Second, we show that the dense model theorem is false even when we make the stronger
assumption that the starting distribution $\bf{Z}$ is dense in a pseudorandom set.
Specifically, in \Cref{thm:dsopstodm} we can show that some distributions $\bf{Z}$ are dense in a pseudorandom 
set but fail to have a dense model when $\cF$ consists of constant-depth, polynomial 
size circuits.

Having contextualized our work some, we now turn to describe our contributions in more detail.

\subsection{Contributions}

We separate the previously described notions of computational entropy, giving examples
where the dense model theorem is false. We are able to prove different separations 
when $\cF$ is  constant-depth unbounded fan-in circuits, 
low-degree polynomials over a finite field, 
and, in one case, any test class $\cF$ which cannot efficiently approximate majority
(in some sense made explicit later on). The only known separation prior was between pseudodensity 
and having a dense model for bounded-width read-once branching programs, due 
to Barak et al. \cite{barak2003computational}. 
%
%

Let $\cC(S, d)$ denote the class of unbounded fan-in, size $S$, depth $d$ circuits. We are generally
thinking of $S = \poly(n)$ and $d = O(1)$, which corresponds to the complexity class $\AC^0$. 
\Cref{thm:dsopstodm} shows that $\bf{Z}$ being $\delta$-dense in an $\eps$-pseudorandom set
need not imply that $\bf{Z}$ has a $\delta$-dense $\eps$-model when the test class is $\cC(S, d)$:

\begin{restatable}{thm}{dsopstodm}
\label{thm:dsopstodm}
    Let $\eps, \eps' > 0$ be arbitrary, $\delta \geq \eps'/8$ and  
    \[
	    S \leq \exp(O\Big(\frac{\sqrt{\eps'}}{\eps} 
	    \cdot \frac{\sqrt{\log(1/\delta)}}{\log(1/\eps')}\Big)^{1/(d-1)}).
    \]
    Then for $\cF = \cC(S, d)$, there is a random variable $\bf{D}$ over $\{ 0, 1\}^n$
    with $n = O(\log(1/\delta)/\eps^2)$
    so that $\bf{D}$ is $\delta$-dense in an $\eps'$-pseudorandom
    set but does not have a $\delta$-dense $\eps$-model.
    In particular, the dense model theorem is false in this setting.
\end{restatable}

Recall that the dense model theorem is 
false when $\eps' = \Omega(\eps \delta)$,
which makes the restriction $\delta \geq \eps'/8$ extremely mild. 
A common regime is $\eps = 1/\poly(n)$, $\delta = O(1)$ 
and $\eps' = \delta \eps = \Theta(\eps)$, in which case this gives us 
(essentially) a lower bound of 
weakly exponential in $1/\sqrt{\eps} \approx 1/\sqrt{\eps'}$. 

Let $\bf{N}_\alpha$
denote the product distribution of $n$ Bernoulli random 
variables with success probability $1/2 - \alpha$. Recall that density 
in a pseudorandom set readily implies pseudodensity, and one 
can use the dense model theorem to show the converse.
We show that $(\eps, \delta)$-pseudodensity need not imply 
$\delta$-density in an $\eps$-pseudorandom set when the test class is $\cC(S, d)$:

\begin{restatable}{thm}{pdtodsops}
\label{thm:pdtodsops}
	Fix $\eps, \eps', \delta > 0$, $d \in \N$, and
	\[
		S \leq \exp(O\Big(\frac{\sqrt{\delta}}{\sqrt{\eps}} \cdot
		\frac{\log(1/\delta)}{\log(1/\eps')}\Big)^{1/(d-1)}).
	\]
	Then $\bf{N}_{\sqrt{\eps/\delta}}$ over $\{0, 1\}^n$ 
	with $n = O(1/\eps)$
	is $(\eps', \delta)$-pseudodense and 
	yet $\bf{N}_{\sqrt{\eps/\delta}}$ is not $\delta$-dense inside of any 
	$\eps$-pseudorandom set.
\end{restatable}

The dependence $\eps'$ means that we can take $\eps'$ exponentially smaller
than $\eps$ and still obtain a separation. This case corresponds to $\cF$ being `very' fooled
by $\bf{N}_\alpha$ but still not being $\delta$-dense in a `mildly' pseudorandom set. 
This result draws on a recent line of work in the pseudorandomness literature
--- often referred to as `the coin problem' and studied in, e.g., \cite{shaltiel2010hardness, cohen2014two, aaronson2010bqp, tal2017tight} --- which 
concerns the ability of a test class $\cF$ unable to compute majority
has in distinguising $\bf{N}_\alpha$ and $\bf{U}$. We will discuss this 
connection in more detail during the proof overviews.  

We prove a similar separation for degree-$d$ $\F_p$-polynomials
(on $n$ variables), which generalizes (and uses techniques from) 
a recent result of Srinivasan
\cite{srinivasan2020} in the case where $\delta = 1$. 
In this case, we think of a distribution $\bf{Z}$
as being $(\eps', \delta)$-pseudodense for degree-$d$ $\F_p$-polynomials
when $\delta \Pr[P(\bf{Z}) \neq 0] - \eps' \geq \Pr[P(\bf{U}) \neq 0]$
for any degree-$d$ polynomial $P \in \F_p[X_1, ..., X_n]$ (noting
that we are only evaluating $P$ over $\{0, 1\}^n$).

\begin{restatable}{thm}{degreelb}
	\label{thm:degreelb}
	Fix a finite field $\F$ with characteristic $p = O(1)$
	, $\eps, \eps' > 0$ and let $c > \delta > 0$
	where $c \approx 1/200$ is an absolute constant. Suppose that
	\[
		d \leq O(\sqrt{\delta/\eps})
	\]
	Then when $\cF$ is the $n$-variate degree-$d$ polynomials over $\F$
	with $n = 1/\eps$, and $\alpha = O(\sqrt{\eps/\delta})$, 
	$\bf{N}_\alpha$ is $(\eps', \delta)$-pseudodense
	but is not $\delta$-dense inside of an $\eps$-pseudorandom set. 
\end{restatable}

This implies lower bounds for constant-depth circuits with $\MOD_p$ gates by the 
classical lower bounds of
Razborov \cite{razborov1987lower} and Smolensky \cite{smolensky1993representations}.
Perhaps more interestingly, this holds even over non-prime fields. 
Also notably, there is no dependence on $\eps' \leq \eps \delta$, so we can take it to be 
arbitrarily small.

We also prove a more general separation between pseudodensity and density 
in a pseudorandom set. This result, drawing from the work of 
\cite{shaltiel2010hardness}, provides a more specific 
characterization of the sense 
in which dense model theorems are `required' to compute majority. 

\begin{restatable}{thm}{pdtodm}
    \label{thm:pdtodm}
    Let $\eps, \delta > 0$.
     Suppose $\cF$ is a test class of boolean functions $f: \{0, 1\}^n \to \{0, 1\}$
	with the following property: there is no $\AC^0$ $\cF$-oracle circuit
	of size $\poly(n \cdot \frac{\sqrt{\delta}}{\eps^{3/2}})$ 
	computing majority on $O(\sqrt{\delta/\eps})$ bits.

	Then $\bf{N}_{\sqrt{\eps/\delta}}$ is $(\epsilon \delta, \delta)$-pseudodense
	and yet does not have a $\delta$-dense $\eps$-model. In particular,
	when the hypotheses are met, the dense model theorem is false.
\end{restatable}

Informally, this says that any $\cF$ which can
refute the pseudodensity of $\bf{N}_\alpha$ is only 
`a constant-depth circuit away' from computing majority. 
\subsection{Related work}
\label{ssection:related_work}

\paragraph{Computational entropy}

Computational entropy was studied systematically in \cite{barak2003computational}
and is relevant to various problems in complexity and cryptography such as
leakage-resilience \cite{dziembowski2008leakage}, 
constructions of PRGs from one-way functions 
\cite{haastad1999pseudorandom, haitner2009inaccessible, haitner2013efficiency}. 
and derandomization \cite{doron2019nearly}. 

There are a number of definitions of computational entropy which we \textit{don't}
consider in this work. For example, Yao pseudoentropy \cite{yao1982theory} (see also \cite{barak2003computational}), corresponding to random
variables which are `compressible' by a class of tests $\cF$, in the sense that $\cF$
can encode and decode the random variable by encoding into a small number of bits.
Yao pseudoentropy was recently used in time-efficient hardness-to-randomness tradeoffs \cite{doron2019nearly},
where (randomness-efficient) samplers for pseudodense distributions were used with an 
appropriate extractor to construct a pseudorandom distribution. 
Another example is \textit{inaccessible entropy} of Haitner et al. \cite{haitner2009inaccessible},
corresponding to the entropy of a message at some round in a two-player protocol 
conditioned on the prior messages and the randomness of the players, which is used
in efficient constructions of statistically hiding commitment schemes 
from one-way functions \cite{haitner2013efficiency}.

Separating notions of computational entropy has been studied before 
in \cite{barak2003computational}, who prove a separation of 
pseudodensity and having a dense model for bounded-width 
read-once branching programs. Separating notions of \textit{conditional} computational entropy
was studied in \cite{hsiao2007conditional}, showing 
separations between conditional variants
of Yao pseudoentropy and having a dense model.

As mentioned in \cite{hsiao2007conditional}, citing \cite{trevisan2009regularity}
and personal communication with Impagliazzo, another question of interest is whether
Yao pseudoentropy (corresponding to efficient encoding/decoding algorithms)
implies having dense model. It is not hard to see that small Yao pseudoentropy
implies small pseudodensity, with some mild restrictions on $\cF$. It would be
interesting to see if the techniques from this paper can be used to understand
Yao pseudoentropy in more detail. We leave this to future work. 

\paragraph{Complexity of dense model theorems and hardness amplification}

Prior work on the complexity of dense model theorems has
included a tight lower bound on the query complexity \cite{zhang2011query}
and a lower bound on the advice complexity \cite{watson2015advice}. 
As far as we are aware, this is the first work to consider the computational complexity 
of dense model theorems.

There has also been prior work on the computational complexity of hardness amplification, establishing that various
known strategies for hardness amplification require the computation of majority \cite{lu2011complexity, shaltiel2010hardness, grinberg2018indistinguishability,shaltiel2020possible}. It is known that a particular type of hardness amplification given by the \textit{hardcore lemma} implies the dense model theorem \cite{impagliazzo2020connections}. 

Our results are stronger in the following sense:
previous work \cite{lu2011complexity, shaltiel2010hardness, 
grinberg2018indistinguishability} shows that \textit{black-box hardness amplification proofs}
require majority. This means that if you amplify the hardness of $f$ in some black-box way, 
then this can be used to compute majority. In our case, we simply show (in different settings)
that the dense model theorem is \textit{false}, regardless of how we tried to prove it.
By the connection between the hardcore lemma and the dense model theorem, our results
also provide scenarios where the hardcore lemma is false. As far as we are aware,
these are the first such scenarios recorded in the literature. 

\subsection{Technical overview}

\label{ssection:technical_overview}
We discuss two general themes that appear consistently in the proofs
and then discuss each of the main theorems in some more detail.

\subsubsection{Dense distributions have mostly unbiased bits}

A commonly-used observation in theoretical computer science is that
most bit positions of a $\delta$-dense random variable over $\{0, 1\}^n$ 
have bias $O(\sqrt{\log(1/\delta)}/n)$
(see, for example, the introduction of \cite{meir2019prediction}).
Relevant to our purposes, it provides a \textit{necessary}
condition for having a $\delta$-dense $\eps$-model with 
respect to any class $\cF$ containing the projections
$z \mapsto z_i$. $\bf{Z}$ has a $\delta$-dense $\eps$-model,
then most bits of $\bf{Z}$ have bias $\eps + O(\sqrt{\log(1/\delta)}/n)$. In particular, if all of the bits of $\bf{Z}$ have \textit{large} bias, then it can't have a dense model. 

This is used directly in the proof of \Cref{thm:dsopstodm}. In this case,
we construct a distribution $\bf{Z}$ which is $\delta$-dense in a set which 
is $\eps$-pseudorandom for $\AC^0$ but where the each bit is 
noticeably biased away from $1/2$.  

In order to prove separations between pseudodensity and being dense in a pseudorandom
set --- as in \Cref{thm:pdtodsops}, \Cref{thm:degreelb} and \Cref{thm:pdtodm} --- 
we need to consider the bias of larger subsets of variables. 
Considering just two bits is sufficient 
to prove mild concentration bounds on the weight of pseudorandom strings.
This implies that the tails of dense subsets of pseudorandom sets 
should not be too heavy.

\subsubsection{Biased coin distribution}

The \textit{biased coin distribution},
$\bf{N}_\alpha$ over $\{0, 1\}^n$ is the product of $n$ Bernoulli random variables 
with success probability $1/2 - \alpha$. $\bf{N}_\alpha$ has recently garnered
significant interest in the pseudorandomness literature 
(see \cite{agrawal2019coin, cohen2014two, 
tal2017tight, brody2010coin, aaronson2010bqp}). 
Shaltiel and Viola \cite{shaltiel2010hardness} showed that if $f$ is a 
test which $\eps$-distinguishes $\bf{N}_\alpha$ from $\bf{U}$,
then there is a small, constant-depth circuit $C$ with $f$-oracle gates which computes
majority on $O(1/\eps)$ bits. A similar, but qualitatively different,
connection due to Limaye et al \cite{limaye2019fixed} --- extended
to any choice of $\eps > 0$
by Srinivasan \cite{srinivasan2020} --- shows that any $\F_p$-polynomial
with advantage $1-2\eps$ in distinguishing $\bf{N}_\alpha$
from $\bf{U}$ must have degree $\Omega(\log(1/\eps)/\alpha)$.
We extend some of these pseudorandomness results regarding $\bf{N}_\alpha$ to \textit{pseudodensity} results. 

First, we extend the observation of Shaltiel and Viola to apply
to tests $f$ for which $\E[f(\bf{Z})] \geq \delta \E[f(\bf{U})] + \eps$
(which corresponds to pseudorandomness when $\delta = 1$). This
gives us unconditional pseudodensity for test classes $\cF$ which
can't be used in small, constant-depth oracle circuits approximating 
majority. We also extend the observation of \cite{limaye2019fixed} to show
lower bounds on the $\F_p$-degree for any function $f$ which refutes
the pseudodensity of $\bf{N}_\alpha$. 

In \Cref{lemma:thm2_pseudodensity},
we show that $\bf{N}_\alpha$ exhibits $(\eps, \delta)$-pseudodensity
for $\eps = (p \cdot O(\log S)^{d-1})^k$ and $\delta = e^{-\alpha k / p}$.
This can be seen as a generalization of Tal's result, building on \cite{cohen2014two, aaronson2010bqp, shaltiel2010hardness} that $\bf{N}_\alpha$
is $3\alpha \cdot O(\log S)^{d-1}$-pseudorandom for $\cC(S, d)$. 

Tal uses a Fourier analytic proof which becomes
very simple given tail bounds on the Fourier spectrum of $\AC^0$ (the 
latter being the main contribution of \cite{tal2017tight}). More generally,
any $\cF$ enjoying sufficiently strong tail bounds on the Fourier spectrum
(in the $\ell_1$ norm) cannot distinguish between $\bf{N}_\alpha$ and uniform.
It turns out, as proved by Tal and recorded in Agarwal \cite{agrawal2019coin},
that if $\cF$ is closed under restrictions than even bounding the first level
of the Fourier spectrum works. The proof of \Cref{lemma:thm2_pseudodensity} based specifically on the switching lemma for constant-depth circuits. While switching lemmas can be used to show Fourier concentration, it would be intersting to find a proof which only uses the assumption of Fourier concentration (or some Fourier-analytic assumption).

\subsubsection{\Cref{thm:dsopstodm}}

Our goal is to construct a random variable $\bf{D}$ which is dense inside of an
$\AC^0$-pseudorandom set but where each bit is biased away from $0$. In this case, $\bf{D}$
would be distinguishable from any dense set, since the average bit of a dense set is roughly unbiased. 
Doing so requires two steps.

The first step is constructing an appropriate distribution $\bf{Z}$ that fools $\AC^0$ circuits.
For this we adopt a general strategy of Ajtai and Wigderson \cite{ajtai1985deterministic} (and applied in many contexts in pseudorandomness since; see, e.g., \cite{servedio2018improved}): 
to fool a circuit $C$,
we start by producing a 
random restriction to simpify $C$ to a short decision tree (via the switching
lemma), and then we fool the decision tree on the remaining bits using a $k$-wise 
independent distribution $\bf{S}$. If we wanted $\bf{Z}$ to have small
support size, we would need some way of producing random restrictions 
with a small amount randomness (which is precisely the approach of Ajtai-Wigderson and later work). 
Fortunately, we only care about the existence of $\bf{Z}$ and are 
therefore content to use the `non-derandomized' switching lemma. 

The second step is finding a dense subset $\bf{D}$ of $\bf{S}$ with biased bits. 
We do this by constructing $\bf{S}$ so that each bit has bias roughly $\sqrt{\log(1/\delta)/K}$,
where $k \ll K \ll n$ is a parameter. This is achieved by randomly bucketing the indices
into $K$ buckets and assigning each bucket a random bit, which reduces
the dimension of the problem from $n$ to $K$. This means we can pick a $\delta$-dense 
event in $\{0, 1\}^K$ with extremal bias --- met (up to constants)
by the function accepting all strings with weight less than 
$K/2 - K \sqrt{\log (1/\delta)}$ --- in order to find a dense subset of $\bf{S}$ with large bias. 
The bucketing construction introduces some error when a small set $I \subseteq [n]$ hits
to distinct elements in some buckets.

\subsubsection{\Cref{thm:pdtodsops}}

We will show $\bf{N}_\alpha$ has $(\delta, \eps')$-pseudodensity
for $\AC^0$ for $\delta = \eps' =  O(1)$, $\alpha = 1/\poly \log(n)$.  The idea is that $\bf{N}_\alpha$ can be sampled by first sampling a random restriction which leaves a $p$ fraction
of the bits unset (and is unbiased on the restricted bits) and then setting
the remaining bits with bias $\alpha/p$. Applying the switching lemma,
we conclude that $\E[f(\bf{N}_\alpha)] \approx \E[f'(\bf{N}_{\alpha/p})]$
where $f'$ is a short decision tree (which doesn't
not depend on all of its inputs). A simple calculation reveals that acceptance probability of $f'$
can increase by at a most a factor $(1 + \alpha/p)^d \leq e^{\alpha d / p}$
when passing from the uniform distribution to $\bf{N}_{\alpha/p}$. By incorporating the error from the switching lemma (i.e. the advantage lost
by conditioning on the switching lemma succeeding), we get $(\delta, \epsilon)$-pseudodensity.

To prove the separation, we use the fact that the Hamming weight of a random variable 
fooling $\cC(S, d)$ is concentrated around its expectation. This means in particular
that if $\bf{N}_\alpha$ \textit{were} $\delta$-dense 
in a pseudorandom distribution, then
the tails of $\bf{N}_\alpha$ couldn't be too heavy and therefore $\alpha$ couldn't be too large.

\subsubsection{\Cref{thm:degreelb} and \Cref{thm:pdtodm}}

\Cref{thm:degreelb} and \Cref{thm:pdtodm} 
draw from related work of Srinivasan \cite{srinivasan2020} 
and Shaltiel-Viola \cite{shaltiel2010hardness} respectively.

With $\epsilon > 0$ and $\cF$ an arbitrary class of tests $f: \{0, 1\}^n \to \{\pm 1\}$,
suppose that $f \in \cF$ witnesses that $\bf{N}_\eps$ \textit{fails} to have
$(\eps', \delta)$-pseudo-density in the sense that
\[
	\E[f(\bf{U})] \leq \delta \E[f(\bf{N}_\beta)] - \gamma.
\]

\cite{srinivasan2020} and \cite{shaltiel2010hardness}
both make use of the following simple observation. 
Given two strings $u, v \in \{0, 1\}^m$ with $\wt(u) = (1/2 - \eps)m$ 
and $\wt(v) = m/2$,
a uniformly random index 
$i \in [m]$ has $u_i$ distributed as a $(1/2 - \eps)$-biased coin
and $v_i$ as an unbiased coin. 
In our case, applying $f$ to sufficiently many random samples from
$u$ or $v$ `distinguishes' the two of them, but in a weaker sense.

In the case of \Cref{thm:pdtodm}, we can amplify
acceptance probabilities by increasing the size of the circuit
by a factor $1/\eps \delta$, after which we can apply 
\cite{shaltiel2010hardness} saying that constant-error
distinguishers between $\bf{N}_\alpha$ and $\bf{U}$
can be used to compute majority. 

For \Cref{thm:degreelb}, we apply a beautiful recent result of Srinivasan 
\cite{srinivasan2020} showing that any $m$-variate polynomial (over a finite field) 
which vanishes on most points on the slice $1/2 - \alpha$ and doesn't 
vanish on most points on the slice $1/2$
must have high degree $\Omega(\alpha m)$. 
One way of interpreting this result is that low-degree
polynomials can't approximately solve certain `promise' versions of majority.

In this latter case, we need to open up the error reduction procedure
we use for \Cref{thm:pdtodm} and show how to approximate it 
using low-degree polynomials. This will ultimately be achieved 
by approximating OR with a probabilistic polynomial, 
as in \cite{razborov1987lower, smolensky1993representations}.

\section{Technical tools}
\label{section:prelim}

We write $[n] = \{1, ..., n\}$ and use boldface to denote random variables. Let $\cC(S, d)$ be the set of size $S$, depth-$d$ unbounded fan-in circuits. For a boolean function
$f: \{0, 1\}^n \to \{0, 1\}$, let $DT(f)$ denote the 
depth of the shortest decision tree computing $f$. 
%
%
%

\subsection{Biased coins}

As before, let $\mathbf{N}_\alpha$ denote the random variable 
corresponding to the product of $n$ independent coins with bias $(1/2 - \alpha)$. 
That is,
\[
	\Pr[\bf{N}_\alpha = z] = (1/2 - \alpha)^{\wt(z)}(1/2 + \alpha)^{n - \wt(z)}
\]
where $\wt(z)$ denotes the Hamming weight of $z$. 

For a random variable $\bf{Z}$ over $\{0, 1\}^n$ and $i \in [n]$,
let $\bias_i(\bf{Z}) = |\Pr[\bf{Z}_i = 1] - \Pr[\bf{Z}_i = 0]|/2$.
Let $\cB = \{z \mapsto z_i : i \in [n] \}$ be the set of monotone projections.
A random variable $\bf{Z} = (\bf{Z}_1, ..., \bf{Z}_n)$ is $\epsilon$-pseudorandom
with respect to $\cB$ precisely when each marginal $\bf{Z}_i$ 
has the property that $\bias_i(\bf{Z}) = |\Pr[\bf{Z}_i = 1] - 1/2| \leq \eps$
for each $i \in [n]$. In particular,

\begin{claim}
	\label{claim:projections}
	For any $\eps > 0$, $\bf{N}_{\eps}$ is $\eps$-pseudorandom with respect to $\cB$. 
\end{claim}

\subsection{Information theory}

The \textit{(Shannon) entropy} of a random variable is defined
as 
\[
H(\bf{Z}) = - \sum_{x \in \{0, 1\}^n} p_{\bf{Z}}(x) \log p_{\bf{Z}}(x),
\]
where $p_{\bf{Z}}$ is the probability density function corresponding to $\bf{Z}$.
The Shannon entropy of random vector is sub-additive, 
in that $H(\bf{Z}) \leq \sum_{i \in [n]} \bf{Z_i}$.
When $\bf{Z} \in \{0, 1\}$ and $\Pr[\bf{Z} = 1] = p$, 
we use $h(p) = H(\bf{Z}) = -(p \log p + (1-p)\log(1-p))$
to denote the binary entropy function.

The \textit{min-entropy} is defined as
\[
	H_\infty(\bf{Z}) = - \min_{x \in \{0, 1\}^n} \log p_{\bf{Z}}(x)
\]
If $\bf{Z}$ is $\delta$-dense inside of $\bf{U}$,
then its min-entropy is $n - \log(1/\delta)$
and for any random variable $\bf{Z}$, 
$H_\infty(\bf{Z}) \leq H(\bf{Z})$.

By this latter inequality and subadditivity, the average entropy of $\bf{Z}$'s
bits is at least $1 - \log(1/\delta)/n$. Appealing to a
quadratic approximation of binary entropy, we learn
that the bias must be at most $\sqrt{\log(1/\delta)/n}$.
This result has been referred to as \textit{Chang's inequality}
and the \textit{Level-1 inequality}, having been observed
in different forms and with different proofs in,
for example, 
\cite{talagrand1996much, chang2002polynomial, hambardzumyan2020chang, impagliazzo2014entropic}.
Because it is so simple, we provide a proof here:

\begin{claim}
    	\label{lemma:chang}
	If $\bf{Z}$ is $\delta$-dense in $\bf{U}$, 
	then $\E_i[\bias_i(\bf{Z})] \leq \sqrt{\log(1/\delta)/n}$
\end{claim}

\begin{proof}
	As $\delta$-density is equivalent to $n - \log(1/\delta)$ min-entropy,
	\[
		n - \log(1/\delta) 
		= H_\infty(\bf{Z}) \leq H(\bf{Z}) \leq \sum_{i \in [n]} H(\bf{Z}_i),
	\]
	by subadditivity of entropy.
	The entropy of $\bf{Z}_i$'s bits, therefore, is at least $1 - \log(1/\delta)/n$ on average.
	Taking the Taylor series, we can approximate the binary entropy function
	$h(p)$ around $1/2$
	by a quadratic function as $h(1/2 + \eps) \leq 1 - (2 / \ln 2) \eps^2$.
	Comparing this bound with the average, we get
	\[
		1 - \log(1/\delta) \leq 1 - (2 / \ln 2) \eps^2,
	\]
	meaning $\eps \leq \sqrt{(\ln 2 / 2) \cdot (\log(1/\delta) / n)} \leq 
	\sqrt{\log(1/\delta)/n}$.
\end{proof}

\subsection{Random variables lacking computational entropy}

It follows directly from \Cref{lemma:chang} that 
if $\bias_i(\bf{Z}_i)$ exceeds $\eps + \sqrt{\log(1/\delta)}/n$
for every $i$, then $\bf{Z}$ does not have a $\delta$-dense $\eps$-model with respect 
to the projections $\cB$.

\begin{lemma}
    \label{lemma:no_dense_models}
    Let $\bf{Z}$ be a random variable
    with $\bias_i(\bf{Z}) \leq \gamma$ 

    for every $i \in [n]$.
    Then for any $\delta > 0$ and 
    $\gamma \geq \epsilon + \sqrt{\frac{ \log(1/\delta)}{n}}$,  
    $\bf{Z}$ does not have a $\delta$-dense $\epsilon$-model
    with respect to $\cB$.
\end{lemma}

This is used for the separation in \Cref{thm:dsopstodm}.
We would also like a necessary condition for being dense in a pseudorandom
set. Towards this end, we note that pseudorandom distributions for even
very simple test classes have mild concentration properties.

\begin{claim}
	\label{claim:concentration}
	Suppose $\cF$ can compute $x_i \oplus x_j$ for every $i, j \in [n]$ and
	let $\bf{Z}$ over $\{0, 1\}^n$ be $\eps$-pseudorandom for $\cF$. 
	Then
	\[
		\Pr[\sum_{i} \bf{Z_i} \leq n/2 - \alpha n] 
		\leq \frac{1}{4\alpha^2 n} + \frac{\eps}{4\alpha^2}
	\]
\end{claim}

\begin{proof}
	We work over $\{\pm 1\}$ instead of $\{0, 1\}$
	to make calculations easier.
	We can compute the second moment as
	\[
		\E[(\sum_i\bf{Z}_i)^2] 
		= \sum_i \E[\bf{Z}_i^2] + \sum_{i \neq j} \E[\bf{Z}_i \bf{Z}_j]
		\leq n + \eps n^2.
	\]
	Applying Markov's inequality to $(\sum_i \bf{Z}_i)^2$, we see that
	\[
		\Pr[|\sum_i \bf{Z}_i| \geq 2\alpha n]
		= \Pr[(\sum_i \bf{Z}_i)^2 \geq (2\alpha n)^2] 
		\leq \E[(\sum_i\bf{Z}_i)^2]/(2\alpha n)^2.
	\]
	We use $2 \alpha n$ because it maps back to $n/2 - \alpha n$
	in $\{0, 1\}$.
	Then the conclusion follows from our second moment calculation
	and converting back to $\{0, 1\}$. 
\end{proof}

The tails of a dense subset can't be too much larger than the original distribution,
by definition of density. This gives us a test for being dense in a pseudorandom set,
which we specialize to $\bf{N}_\alpha$.

\begin{lemma}
	\label{lemma:coin}
	Let $\eps, \delta > 0$ be arbitrary.
	Suppose $\cF$ can compute $x_i \oplus x_j$ for any $i, j \in [n]$
	and $\alpha \geq \sqrt{1/(8\delta) \cdot (1/n + \eps)}$.
	Then $\bf{N}_\alpha$ is not $\delta$-dense in any set which is 
	$\eps$-pseudorandom for $\cF$.
\end{lemma}

\begin{proof}
	Under $\bf{N}_\alpha$, the volume of the threshold 
	$\bf{1}[\sum_i \bf{Z}_i \leq n/2 - \alpha n]$ is $1/2$.
	Taking \Cref{claim:concentration} in the contrapositive, 
	we reach the desired
	conclusion when 
	\begin{align*}
		1/2 &> \frac{1}{4\delta \alpha^2 n} + \frac{\eps}{4\delta \alpha^2}\\
		\alpha^2 &> \frac{1}{8\delta}(1/n + \eps)\\
	\end{align*}
\end{proof}

\subsection{Random restrictions and the switching lemma}

A restriction over $[n]$ is a function $\rho: [n] \to \{0, 1, *\}$. Indices in $\rho^{-1}(*)$ can be thought of as \textit{unset} and each other index as \textit{set}.  
For another restriction $z$ so that $\rho^{-1}(*) \subseteq z^{-1}(\{0, 1\})$, let $\rho \circ z \in \{0, 1\}^n$ be defined by
\[
    (\rho \circ z)_i = \begin{cases}
        z_i \text{ if $i \in \rho^{-1}(*)$,}\\
        \rho_i \text{ otherwise.}
    \end{cases}
\]
Define the restricted function $f|_{\rho} : \{0, 1\}^{\rho^{-1}(*)} \to \{0, 1\}$ over $\rho$'s unset indices by
\[
    f|_{\rho}(z) = f(\rho \circ z).
\]

Let $R_p$ be the distribution on restrictions over $[n]$ obtained by setting
$\rho(i) = *$ independently with probability $p$, and then setting
each bit not assigned to $*$ a random bit.
The switching lemma we use is due to Rossman \cite{rossman2017entropy}, 
building on a long line of work \cite{ajtai1985deterministic, haastad1987computational, haastad2014correlation, impagliazzo2012satisfiability}:

\begin{thm}[Rossman \cite{rossman2017entropy}]
    \label{thm:switching_lemma}
	Suppose $f \in \cC(S, d)$.
	Then 
	\[
		\Pr_{\bs{\rho} \sim R_p}[DT(f|_{\bs{\rho}}) \geq k] \leq (p \cdot O(\log S)^{d-1})^k
	\]
\end{thm}

By considering a random restriction $\bs{\rho} \sim R_p$ over $[n]$ and a random variable $\bf{Z}$ over $\{0, 1\}^{n}$,
the definition of a restricted function implies that
\[
    \E[f(\bs{\rho} \circ \bf{Z})] = \E[f|_{\bs{\rho}}(\bf{Z})].
\]

We make crucial use of two simple corollaries of the switching lemma,
which allow us to reason about distinguishability for
$\AC^0$ circuits in terms of distinguishability for short decision
trees. 

\begin{lemma}
    \label{lemma:first_sl_lemma}
    Suppose $f \in \cC(S, d)$. Then there is a distribution over depth $k$ decision trees so that
    \[
        |\E[f(\bs{\rho} \circ \bf{Z})] - \E[h_{\bs{\rho}}(\bf{Z})]| \leq (p \cdot O(\log S)^{d-1})^k
    \]
\end{lemma}

\begin{proof}
    Let $g_{\bs{\rho}}$ denote the optimal decision tree for $f|_{\bs{\rho}}$. Let $E$ denote the event that $g_{\bs{\rho}}$ has depth at most $k$ and $\Pr[E] = 1 - q$. Let
    $h_{\bs{\rho}}$ be the distribution over depth at most $k$ 
    decision trees obtained by sampling $g_{\bs{\rho}}$ conditioned on $E$. Then
    \begin{align*}
        \E[f(\bs{\rho} \circ \bf{Z})] &= \E[f|_{\bs{\rho}}(\bf{Z})]\\
        &= (1 - q) \E[g_{\bs{\rho}}(\bf{Z}) | E ] + q \E[g_{\bs{\rho}}(\bf{Z}) | \neg E ]\\
        &= (1 - q) \E[h_{\bs{\rho}}(\bf{Z})] + q \E[g_{\bs{\rho}}(\bf{Z}) | \neg E ]\\
        &= \E[h_{\bs{\rho}}(\bf{Z})] - q(\E[h_{\bs{\rho}}(\bf{Z})] - \E[g_{\bs{\rho}}(\bf{Z}) | \neg E ]).
    \end{align*}
    The right-hand term is bounded in absolute value by $q$ because $f$ is Boolean.  By \Cref{thm:switching_lemma}, $q \leq (p \cdot O(\log S)^{d-1})^k$. 
\end{proof}

\begin{lemma}
    \label{lemma:second_sl_lemma}
    Suppose $f \in \cC(S, d)$. Then there's a depth $k$
    decision tree $h$ so that
    \[
        |\E[f(\bf{U})] - \E[f(\bs{\rho} \circ \bf{Z})]| \leq  |\E[f'(\bf{U})] - \E[f'(\bf{Z})]| + (p \cdot O(\log S)^{d-1})^k
    \]
\end{lemma}

\begin{proof}
    \Cref{lemma:first_sl_lemma} gives us the following upper bound.
    \begin{align*}
        |\E[f(\bf{U})] - \E[f(\bs{\rho} \circ \bf{Z})]| 
        &\leq |(\E[h_{\bs{\rho}}(\bf{U})] \pm q) - (\E[h_{\bs{\rho}}(\bf{Z})] \pm q)| &\text{(\Cref{lemma:first_sl_lemma})}\\
        &\leq |\E[h_{\bs{\rho}}(\bf{U})] - \E[h_{\bs{\rho}}(\bf{Z})]| + 2q &\text{(triangle inequality)}
    \end{align*}
    We can continue to upper bound the right-hand term by 
    \begin{align*}
         |\E[h_{\bs{\rho}}(\bf{U})] - \E[h_{\bs{\rho}}(\bf{Z})]| &=  |\E_{\bs{\rho}}[\E[h_{\rho}(\bf{U})] - \E[h_{\rho}(\bf{Z})]]|\\
         &\leq \E_{\bs{\rho}}[|\E[h_{\rho}(\bf{U})] - \E[h_{\rho}(\bf{Z})]|] &\text{(triangle inequality)}\\
         &\leq |\E[h(\bf{U})] - \E[h(\bf{Z})]|
    \end{align*}
    where the last line holds for some $h$ in the support of 
    $h_{\bs{\rho}}$ by averaging. 
\end{proof}

\section{Proof of \Cref{thm:dsopstodm}}

We start by reducing the problem of constructing a pseudorandom
$\bf{Z}$ for $\AC^0$ to constructing a pseudorandom $\bf{Z}$
for small-depth decision trees. This can be immediately achieved
by applying \Cref{lemma:second_sl_lemma}.

\begin{claim}
\label{claim:reduction}
Let $p \in [0, 1]$ be arbitrary and suppose $\bf{Z}$ is a random variable over $\{0, 1\}^{n}$ which is $\epsilon$-pseudorandom for depth-$k$ decision trees. Then for $\bs{\rho} \sim R_p$, $\bs{\rho} \circ \bf{Z}$ is $\epsilon'$-pseudorandom for $\cC(S, d)$
for  
\[
    \epsilon' = \epsilon + (p \cdot O(\log S)^{d-1})^k
\]
\end{claim}

The next lemma constructs a pseudorandom distribution
for depth-$k$ decision trees with each bit having significant bias.
\begin{lemma}
	\label{lemma:kwise}
	For any $k \in \N, \delta > 0$ and $K \geq 1/2\delta$, there is
	a $k$-wise independent random variable $\bf{S}$ over 
	$\{0, 1\}^n$ and a $\delta$-dense subset $\bf{D}$ of 
	$\bf{S}$ with the property that 
	\begin{enumerate}
		\item $\bf{D}$ is $\delta$-dense in $\bf{S}$.
		\item For all $i \in [n]$, $\bias_i(\bf{D}) 
			= \Omega(\sqrt{\log(1/\delta)/8K})$
		\item $\bf{S}$ is $k^2/K$-pseudorandom for depth-$k$ decision trees.
	\end{enumerate}
\end{lemma}

We will use the following standard lower bound on the lower
tail of a binomial distribution:
\begin{claim}[\cite{ash1990information}]
	\label{claim:reverse_chernoff}
	For $0 < \alpha < 1$ and let 
	$\bf{Z}_1, ..., \bf{Z}_K$ be independent unbiased coins ($\{0, 1\}$-valued).
	Then any $\gamma$ with $1/2 - \gamma = r/K$ for some positive integer 
	$r$ satisfies
	\[
		\frac{2^{-K (1 - h(1/2 - \gamma))}}
		{\sqrt{2K}}
		\leq \Pr[\sum_{i \in [K]} \bf{Z}_i \leq K/2 - K\gamma] 
	\]
\end{claim}

\begin{proof}[Proof of \Cref{lemma:kwise}]
	We sample $\bf{S}$ in two stages. First, randomly partition $[n]$ 
	into $K$ parts $\bf{A}_1, ..., \bf{A}_K$ for $K > k^2$. Second, assign to each
	$A_i$ a uniformly random bit $\bf{b}_i$.

	Let $\bf{D}$ be $\bf{S}$ conditioned on $\bf{b} = (\bf{b}_1, ..., \bf{b}_K)$
	having weight less than $K/2 - \sqrt{K \log(1/\delta)/8}$. 
	Since the $\bf{b}_i$'s are unbiased random bits, we
	can apply \Cref{claim:reverse_chernoff} to lower 
	bound $\bf{D}$'s density: for any $\gamma$,
	\[
		\Pr[ \sum_{i \in [k]} \bf{b}_i 
		\leq \gamma K]
		\geq \frac{2^{-K h(1/2 - \gamma)}}{\sqrt{2 K}}.
	\]
	This is at least $\delta$ when
	\begin{align*}
		\frac{2^{-K(1 - h(1/2 - \gamma))}}{\sqrt{2 K}} 
		&\geq \delta\\
		1 - h(1/2 - \gamma) &\geq \log(1/\delta)/K - \log(2K)/2K\\
		4\gamma^2 &\geq \log(1/\delta)/K - \log(2K)/2K 
	\end{align*}
	with the upper bound in the last line following from 
	$h(1/2 - \gamma) \geq 1 - 4\gamma^2$. Hence, if the set
	of strings with weight at most $K/2 - \gamma K$ 
	is $\delta$-dense, 
	we have $\gamma \geq \frac{1}{2}\sqrt{\log(1/\delta)/K - \log(2K)/2K}$. 
	$\log(2K)/2K$ is at most $\log(1/\delta)/2K$ when $2K \leq 1/\delta$,
	in which case $\gamma \geq \sqrt{\log(1/\delta)/8K}$.
	In particular, this lower bounds the bias of $\bf{D}$'s bits.

	To see why it's $k^2/K$-pseudorandom for depth-$k$ decision trees, consider
	a depth-$k$ decision tree $T$. Over $\bf{U}$, we can
	imagine evaluting $T$ `on-line' as follows: 
	whenever $T$ queries the $i$th bit, determine the value of $z_i$
	by flipping an unbiased coin. Over $\bf{S}$, we can imagine
	evaluating $T$ similarly, where we determine the bucket $A_j$
	that $i$ lives in and the value $b_j$ of that bucket. 

	By conditioning $\bf{S}$ on \textit{not} placing two distinct indices $i, j$
	in the same bucket --- call this conditioned random variable $\bf{S}'$ --- then
	$T$ doesn't have \textit{any} distinguish advantage over $\bf{S}'$,
	as all of the bits it queries are independent and uniform.
	By a union bound, $\bf{S}$ places two distinct 
	indices in the same bucket with probability at most $k^2/K$.
	$T$'s distinguishing advantage is therefore at most $k^2/K$.
\end{proof}

In principle, we could have used other pseudorandom distributions for
decision trees such as the $\eps$-almost $k$-wise independent
distributions from \cite{alon1992simple}. The construction here
is used to obtain better dependence on the parameters of interest.
We will also need a claim to 
expresses the bias of the bits in $\bs{\rho} \circ \bf{Z}$.
The proof can be found in the appendix. 

\begin{restatable}{claim}{biasedbits}
\label{claim:bias_of_rhoZ}
    Fix $p \in [0, 1]$ and a random variable $\bf{Z}$. Let $E$ be an event
    which is independent from $\bf{\rho}$ (in that the conditional distribution
    of $\bf{\rho}$ is identical to the unconditioned distribution). Then
    \[
        \Pr[(\bs{\rho} \circ \bf{Z})_i = 1 | E] = p \Pr[\bf{Z}_i = 1 | E] + (1-p)/2
    \]
\end{restatable}

\Cref{thm:dsopstodm}, which we restate here, is obtained by an appropriate setting of parameters.

\dsopstodm*

\begin{proof}
	Let $n = \log(1/\delta)/\eps^2$, $k = \log(2/\eps')$ and 
	$K = (2 k^2)/\eps'$. We also need $K \geq 1/2\delta$ by the
	restriction in \Cref{lemma:kwise}, which explaines the restriction
	$8 \delta k^2 \geq \eps'$, simplified by using $8 \delta \geq \eps'$
	(a stronger restriction) instead.
	Let $\bf{S}$ and $\bf{D}$ be the random variables from \Cref{lemma:kwise}.
	By \Cref{claim:bias_of_rhoZ}, the bias of $\bs{\rho} \circ \bf{S}$ 
	(where $\bs{\rho} \sim R_p$) is $p \sqrt{\log(1/\delta)/8K}$.
	By \Cref{claim:reduction} and \Cref{lemma:kwise}, $\bs{\rho} \circ \bf{S}$ is 
	$\eps' = k^2/K + (p O(\log S)^{d-1})^k$ pseudorandom. 
	We can also ensure that $\bs{\rho} \circ \bf{S}$
	does \textit{not} have a $\delta$-dense $\eps$-model when 
	$p \sqrt{\log(1/\delta)/8K} \geq \eps + \sqrt{\log(1/\delta)/n}$,
	by \Cref{lemma:no_dense_models}.
	
	By substituting,
	$p \geq 2\sqrt{K/n} = 2\sqrt{(k \eps)^2 / \eps' \cdot \log(1/\delta)}$.
	In comparison, $\eps' \geq k^2/K + (p O(\log S)^{d-1})^k$.
	Recalling that $k^2/K = \eps'/2$, we get that 
	\begin{align*}
		\eps'/2 &\geq (2\sqrt{K/n} O(\log S)^{d-1})^k\\
		\frac{\sqrt{n}}{2 \sqrt{K}} (\eps'/2)^{1/k}
		      &\geq O(\log S)^{d-1}\\
		      \frac{ \sqrt{\log 1/\delta}}{\eps} \cdot 
		      \frac{\sqrt{\eps'}}{2 \sqrt{2}k} \cdot
			(\eps'/2)^{1/k}
		      &\geq O(\log S)^{d-1}.\\
		      \frac{ \sqrt{\log 1/\delta}}{\eps} \cdot 
		      \frac{\sqrt{\eps'}}{\sqrt{32}\log(1/\eps')}
		      &\geq O(\log S)^{d-1}.\\
	\end{align*}
	The claim follows by solving for $S$. 	
\end{proof}

\section{Proof of \Cref{thm:pdtodsops}}

\Cref{thm:pdtodsops} follows by combining \Cref{lemma:coin}
and the following lemma:

\begin{lemma}
\label{lemma:thm2_pseudodensity}
    $\bf{N}_{\alpha}$ has $(\eps, \delta)$-pseudodensity for $\cC(S, d)$ for $\eps = (p \cdot O(\log S)^{d-1})^k$
    and $\delta = e^{- \alpha k / p}$. 
\end{lemma}

Of note, the only additive error depends on the error
from the switching lemma.
Compare this with the claim that 
$\bf{N}_\alpha$ is $(3\alpha \cdot O(\log S)^{d-1})$-pseudorandom
(and therefore has the same pseudodensity for $\delta = 1$)
for $\cC(S, d)$, due to Tal \cite{tal2017tight}.

To prove the lemma, we need a few claims.

\begin{claim}
    \label{claim:thm2_additive_error}
    Suppose $f \in \cC(S, d)$. Then there is a depth-$k$ decision
    tree $h$ with the property that:
    \[
        \E[f(\bf{N}_\alpha)] \leq \E[h(\bf{N}_{\alpha/p})] + (p \cdot O(\log S)^{d-1})^k
    \]
\end{claim}

\begin{proof}
    Take $\bf{Z} = \bf{N}_{\alpha/p}$ in \Cref{lemma:first_sl_lemma}, so we have $\bs{\rho} \circ \bf{N}_{\alpha/p} = \bf{N}_\alpha$ and 
    \[
        \E[f(\bf{N}_\alpha)] \leq \E[h_{\bs{\rho}}(\bf{N}_{\alpha/p})] + (p \cdot O(\log S)^{d-1})^k
    \]
    Averaging over
    $\bf{\rho}$ yields the fixed decision tree. 
\end{proof}

Second, we can upper bound the extent to which the acceptance probability
of a short decision tree increases when passing from the uniform distribution 
$\bf{U}$ to the biased distribution $\bf{N}_\gamma$ 

\begin{restatable}{claim}{densityincrease}
    \label{claim:thm2_multiplicative_error}
    Suppose $f: \{0, 1\}^n \to \{-1, 1\}$ is a depth-$k$ decision tree. 
    Then
    \[
        \E[f(\bf{N}_\gamma)] \leq (1 + \gamma)^k \cdot 
	\E[f(\bf{U})] \leq e^{\gamma k} \cdot \E[f(\bf{U})]
    \]
\end{restatable}

The proof amounts to a calculation, which we include in the appendix. 
We're now in a position to prove the lemma. 

\begin{proof}[Proof of \Cref{lemma:thm2_pseudodensity}]
    Directly applying \Cref{claim:thm2_additive_error}, we get
    \[
           \E[f(\bf{N}_\alpha)] \leq \E[f'(\bf{N}_{\alpha/p})] + (p \cdot O(\log S)^{d-1})^k
    \]
    Applying \Cref{claim:thm2_multiplicative_error} to $\E[f'(\bf{N}_{\alpha/p})]$,
    we get 
    \begin{align*}
	    \E[f(\bf{N}_\alpha)] &\leq (1 + \alpha/p)^k \E[f'(\bf{U})]\\
				 &\leq e^{\alpha k / p} \E[f'(\bf{U})].
    \end{align*}
    Putting these together finishes the proof. 
\end{proof}

We can now prove \Cref{thm:pdtodsops}, restated here:

\pdtodsops*

\begin{proof}[Proof of \Cref{thm:pdtodsops}]
	Let $n = 1/(7\eps)$, $k = \log(1/\eps')$ and $\alpha = \sqrt{\eps/\delta}$.
	These choices satisfy $\alpha \geq \sqrt{\frac{1}{8 \delta}}(1/n + \eps)$,
	meaning $\bf{N}_\alpha$ is not $\delta$-dense in any $\eps$-pseudorandom
	set for $\cC(S, d)$, by \Cref{lemma:coin}.

	By \Cref{lemma:thm2_pseudodensity}, 
	$\bf{N}_\alpha$ has $(\eps', \delta)$-pseudodensity for 
	$\delta = e^{-\alpha k/p}$
	and $\eps' = (p \cdot O(\log S)^{d-1})^k$.

	The constraint on the density implies
	\begin{align*}
		\delta &= e^{-\alpha k / p}\\
		\log(1/\delta) &= \alpha k / p\\
		\log(1/\delta) &= \sqrt{\eps/\delta}\log(1/\eps')/p\\
		p &= \frac{\sqrt{\eps/\delta}\log(1/\eps')}{\log(1/\delta)}.
	\end{align*}
	Plugging this value of $p$ into the expression for $\eps'$, we get
	\begin{align*}
		\eps' &= (p \cdot O(\log S)^{d-1})^k\\
		(\eps')^{1/k}/p &= O(\log S)^{d-1}\\
		(\eps')^{1/\log(1/\eps')} 
		\cdot \frac{\sqrt{\delta}\log(1/\delta)}{\sqrt{\eps}\log(1/\eps')}
				&= O(\log S)^{d-1}.
	\end{align*}
	Note that $(\eps')^{1/\log(1/\eps')} 
	= 2^{- \log(1/\eps')/\log(1/\eps')} = 1/2$. 
	Solving for $S$ gives the claimed bound.
\end{proof}

\section{Proofs of \Cref{thm:degreelb} and \Cref{thm:pdtodm}}

In this section, we prove \Cref{thm:degreelb} and \Cref{thm:pdtodm}.
We include them in the same section due to their similarity
and start with \Cref{thm:degreelb} which is more involved. 
Indeed, \Cref{thm:pdtodm} will follow directly from a
result of Shaltiel and Viola \cite{shaltiel2010hardness}
after an appropriate error reduction procedure.

\subsection{Proof of \Cref{thm:degreelb}}

As in \Cref{thm:pdtodsops},
we will prove unconditional pseudodensity for
$\bf{N}_\alpha$ and compare it with the lower bound
for $\alpha$ given by \Cref{lemma:coin}. 

Let $\bf{Sp}_{n, k}$ denote the (random variable uniform over the) 
set of $n$-bit strings
of weight exactly $n/2 - k$. 
We convert a function which refutes the pseudodensity of $\bf{N}_\alpha$
into a (random) function which distinguishes between the two `slices' of the hypercube
of weight $m/2 - \alpha m$ and $m/2$, which extend similar ideas
from \cite{shaltiel2010hardness, limaye2019fixed, srinivasan2020}.

\begin{lemma}
	\label{lemma:promise_maj}
	Let $f: \{0, 1\}^n \to \{0, 1\}$ be a function
	for which $\delta \E[f(\bf{N}_\alpha)] - \eps' \geq \E[f(\bf{U})]$
	and let $q = \E[f(\bf{N}_\alpha)]$.
	Then for any positive integer $\ell$,
	there is a random function $\bf{F}: \{0, 1\}^m \to \{0, 1\}$
	so that
	\begin{align*}
		\Pr[\bf{F}(\bf{Sp}_{m, \alpha m}) = 0] &\leq e^{-\ell q};\\
		\Pr[\bf{F}(\bf{Sp}_{m, 0}) = 1] &\leq \ell q \delta.
	\end{align*}
	Moreover, $\bf{F}$ is the OR of $\ell$ copies of $f$ (on random inputs). 
\end{lemma}

\begin{proof}
	Fix an input $z \in \{0, 1\}^m$. 
	For $i \in [\ell]$, let $\bf{I}_i$ denote 
	the random length $n$ sequence over $[m]$ obtained 
	by sampling $n$ indices $j_1, ..., j_n \in [m]$ independently,
	uniformly at random (with replacement). 

	The random function is defined as
	\[
		\bf{F}(z) = \bigvee_{i \in [\ell]} f(z_{\bf{I}_i})
	\]
	We can evaluate $\bf{F}$'s acceptance probability on 
	$\bf{Sp}_{m, \alpha m}$ and $\bf{Sp}_{m, 0}$
	as follows:
	\begin{enumerate}
		\item Suppose $z \in \bf{Sp}_{m, \alpha m}$. 
			Then each $z_{\bf{I}_i}$ is distributed as 
			$\bf{N}_\alpha$ on $n$ bits, 
			meaning the $z_{\bf{I}_i}$'s
			constitute $\ell$ independent samples from 
			$\bf{N}_\alpha$. By definition, $f$ 
			accepts with probability $q$ over 
			$\bf{N}_\alpha$ and so the probability that $f$
			doesn't accept in $\ell$ independent runs is 
			$(1-q)^k \leq e^{-kq}$.

		\item Suppose $z \in \bf{Sp}_{m, m/2}$. 
			Then $z_{\bf{I}_i} \in \{0, 1\}$ is a uniformly
			random string, meaning we have $\ell$ independent
			samples from $\bf{U}$. For each sample, $f$ outputs
			$1$ independently with probability $q' = \E[f(\bf{U})]$,
			so that this occurs once in $\ell$ attempts
			happens with probability
			at most $\ell q'$. Since $q' \leq q\delta - \eps \leq q \delta$,
			we can bound this by $\ell q \delta$.
	\end{enumerate}
\end{proof}

\begin{lemma}
	\label{lemma:polynomial_pseudodensity}
	Fix $\eps', \delta > 0$
	and $\delta < c$ for some absolute constant $c \approx 1/200$
	Let $\F$ be a field of positive characteristic $p = O(1)$
	and let $\cF$ be the set of all polynomials $P \in \F[X_1, ..., X_n]$
	with degree at most $d$. Then if $\bf{N}_\alpha$ 
	has $(\eps', \delta)$-pseudodensity with respect to $\cF$, $d = O(1/\alpha)$.
\end{lemma}

Note that there is no dependence on $\eps'$. This is a relic
of the field size $p > 0$, which allows us to approximate large fan-in
ORs with polynomials whose degree does not depend on the fan-in but only
the quality of approximation. 

We will proceed in the contrapositive; if $\bf{N}_\alpha$
\textit{isn't} pseudodense, we will use the witness to construct a
polynomial which we can prove degree lower bounds on directly. By ensuring
that the degrees are related by a constant factor, we get the lower bound.
The proof uses the following approximation of the OR
function as a low-degree polynomial over $\F$ when $\F$ has positive 
characteristic (in characteristic zero, there is dependence on the
fan-in of the OR, which we want to avoid). 

\begin{claim}[\cite{razborov1987lower}]
	\label{claim:razborov}
	For any $n$ and any finite field $\F$ with characteristic $p > 0$,
	there is a distribution $\bf{R}$ on degree $d$ polynomials
	so that for all $z \in \{0, 1\}^n$
	\[
		\Pr[\bf{R}(x) = OR(x)] \geq 1 - \gamma
	\]
	where $d \leq p \log(1/\gamma)$. 
	Moreover,
	$\bf{R}$ is supported on polynomials $R$
	with $R(z) \in \{0, 1\}$ for every $z \in \{0, 1\}^n$
\end{claim}

We also use a special case of the robust Heg\"{e}dus lemma,
discovered recently by Srinivasan \cite{srinivasan2020}.

\begin{lemma}[Robust Heg\"{e}dus lemma (special case), \cite{srinivasan2020}]
	\label{lemma:hegedus}
	Let $\F$ be a finite field.
	Let $2^{-m/100} \leq \lambda \leq c$ where $c < 1$ is a (small)
	absolute constant. 
	Let $\alpha^2 m$ be an integer so that $2^{-2\alpha^2 m} \geq \lambda$.
	Then if $P: \F^n \to \F$ is a degree $d$ polynomial for which:
	\begin{enumerate}
		\item $\Pr[P(\bf{Sp}_{m, \alpha m}) \neq 0] \leq \lambda$
		\item $\Pr[P(\bf{Sp}_{m, 0}) = 0] \leq 1 - e^{-\alpha^2 m/2}$
	\end{enumerate}
	Then $d = \Omega(\alpha m)$.
\end{lemma}

Now we can prove the main lemma.

\begin{proof}[Proof of \Cref{lemma:polynomial_pseudodensity}]
	Let $P: \F^n \to \F$ be a degree $d$ polynomial.
	We can assume $P$ takes boolean values over $\{0, 1\}^n$
	by replacing $P$ with $P^{p-1}$ ($p$ being $\F$'s characteristic).
	This increases the degree to less than $pd$. 
	Assume towards a contradiction that $P$ can certify
	that $\bf{N}_\alpha$ is not $(\eps', \delta)$-pseudodense, in the sense that 
	$\delta \Pr[P(\bf{N}_\alpha) \neq 0] - \eps' \geq \Pr[P(\bf{U}) \neq 0]$.

	Let $q = \Pr[P(\bf{N}_\alpha)]$.
	Applying \Cref{lemma:promise_maj} to $P$ (noting that $P$ is boolean
	on $\{0, 1\}^n$) with $\ell$ specified later, 
	we obtain a random polynomial $\bf{F}$ of the 
	form $\lor_i P(z_{\bf{I}_i})$.
	For a fixed $z \in \{0, 1\}^m$, let 
	$\bf{X}(z) = (z_{\bf{I}_i})_{i \in [\ell]}$
	For $\gamma$ to be determined later, let $\bf{R}$ be the random polynomial
	from \Cref{claim:razborov} and we remark that $\bf{R}$ can be made to output
	boolean values on boolean inputs. Then for any $z \in \{0, 1\}^m$, 
	$\bf{F}(z) = OR(\bf{X}(z))$ so $\bf{F}(z) = \bf{R}(\bf{X}(z))$
	with probability $1 - \gamma$. 	In sum, $\bf{R}(\bf{X}(\cdot))$ 
	has the property that:	
	\begin{align*}
		\Pr[\bf{R}(\bf{X}(\bf{Sp}_{m, \alpha m})) = 0] 
		&\leq e^{-\ell q} + \gamma;\\
		\Pr[\bf{R}(\bf{X}(\bf{Sp}_{m, 0})) = 1] &\leq \ell q \delta + \gamma.
	\end{align*}
	Moreover, $\bf{R}(\bf{X}(\cdot))$ is a random
	polynomial over $\F$ of degree at most $d p^2 \log(1/\gamma)$,
	with the $p^2$ coming from possibly replacing $P$ with $P^{p-1}$.

	We now apply \Cref{lemma:hegedus} to $\bf{R}(\bf{X}(\cdot))$.
	Let $c$ denote the constant from the statement of the lemma.
	Let $m = \ln(1/c)/2\alpha^2$, $\lambda = e^{-\ell q} + \gamma \leq c$, 
	$\gamma = c/2$ and $\ell = \ln(2/c)/q \leq (\ln(2/c)\delta)/\eps'$ where the inequality
	is because $q \geq \eps'/\delta$ follows from $\delta q - \eps' \geq 0$.
	With this setting of parameters, note that,
	in regards to the hypotheses of \Cref{lemma:hegedus},
	we have $2^{-m/100} \leq \lambda \leq 2^{-2\alpha^2 m} = c$,
	
	By calculating, $\bf{R}(\bf{X}(\cdot))$'s classifies $\bf{Sp}_{m, 0}$
	incorrectly with probability $e^{-\ell q} = c/2$
	and classifies $\bf{Sp}_{m, \alpha m}$ incorrectly with probability
	$\ell q \delta \leq \ln(2/c) \delta$, which we can bound away
	from $1$ by choosing $\delta$ to be a sufficiently small constant.

	Therefore, by
	\Cref{lemma:hegedus}, the degree of some polynomial in the support of
	$\bf{R}(\bf{X}(\cdot))$
	is $\Omega(1/\alpha)$. Since $\gamma$ and $p$ ($\F$'s characteristic)
	are constants, the same conclusion therefore holds of $P$.
\end{proof}

We can now prove \Cref{thm:degreelb}, restated here:

\degreelb*

\begin{proof}[Proof of \Cref{thm:degreelb}]
	If $\bf{N}_\alpha$ is $(\eps', \delta)$-pseudodense with respect to
	degree-$d$ polynomials over $\F$, then $d = O(1/\alpha)$.
	On the other hand, taking $\alpha \geq \sqrt{1/8\delta(1/n + \eps)}$
	implies that $\bf{N}_\alpha$ is not $\delta$-dense
	in an $\eps$-pseudorandom set. Therefore, if we take
	$n = 1/\eps$ then picking $\alpha \geq \sqrt{2\eps/8\delta}$
	will give us a separation.
\end{proof}

\subsection{Proof of \Cref{thm:pdtodm}}

In what is becoming tradition, we will show
an unconditional pseudodensity result,
from which \Cref{thm:pdtodm} will follow
by \Cref{lemma:coin}.

\begin{lemma}
	\label{lemma:thm3_main_lemma}
	Let $\eps > 0$ and $1/4 > \delta > 0$ be arbitrary.
	Suppose $\cF$ is a test class of boolean functions $f: \{0, 1\}^n \to \{0, 1\}$
	with the following property: there is
	no $\AC^0$ $\cF$-oracle circuit of size
	$\poly(n/\alpha \eps)$ which computes $\MAJ$ on
	$O(1/\alpha)$ input bits. Then $\bf{N}_\alpha$ is $(\epsilon \delta, \delta)$-pseudodense.
\end{lemma}

This will follow from the result of from Shaltiel and Viola:

\begin{thm}[\cite{shaltiel2010hardness}]
	\label{thm:shaltielviola}
	Let $f: \{0, 1\}^n \to \{0, 1\}$ be a function that distinguishes between $\bf{U}$
	and $\bf{N}_\alpha$ with constant distinguishing probability.
	Then there is an $\AC^0$-circuit of size $\poly(n/\alpha)$
	using $f$-oracle gates which computes majority on $O(1/\alpha)$ bits. 
\end{thm}

\begin{proof}[Proof of \Cref{lemma:thm3_main_lemma}]
	Suppose $\bf{N}_\alpha$ is not $(\epsilon \delta, \delta)$-pseudodense
	for $\cF$ as in the statement of the theorem and let $f$
	witness this fact with $q = \E[f(\bf{N}_\alpha)]$. 

	Let
	\[
		\bf{F}(z) = \bigvee_{j \in \ell} f(\bs{\rho}_i(z))
	\]
where each $\bs{\rho}_i$ is a random permutation of $[n]$.
	When $z \sim \bf{N}_\alpha$, $\bs{\rho}_i(z)$ is distributed
	as $\bf{N}_\alpha$ and likewise for $\bf{U}$. 
	Hence $\bf{F}$ rejects samples from $\bf{N}_\alpha$
	with probability $(1 - q)^\ell \leq e^{- q \ell}$
	which is constant when $\ell = 1/q \leq \delta/\eps' = 1/\eps$,
	this latter property following because $\delta q - \eps' > 0$.
	Additionally, $\bf{F}$ accepts samples from $\bf{U}$
	with probability $\ell q \delta$, which is at most $\delta$
	for our choice of $\ell$. Averaging, some function 
	in $\bf{F}$'s support distinguishes $\bf{N}_\alpha$
	and $\bf{U}$ with constant advantage. Applying \Cref{thm:shaltielviola}
	yields a small $\AC^0$ circuit computing majority, which is a contradiction. 
\end{proof}

\pdtodm*

The proof is the same as \Cref{thm:degreelb}, so we omit the details.

\bibliographystyle{plain}
\bibliography{bib}

\appendix

\section{Omitted proofs}

\subsection{Proof of \Cref{claim:bias_of_rhoZ}}

\biasedbits*

\begin{proof}
Let $R_i$ be the event that $i \in \bs{\rho}^{-1}(*)$. 
	\begin{align*}
	\Pr[(\bs{\rho} \circ \bf{Z})_i = 1 | E] 
	&= \Pr[R_i]\Pr[\bf{Z}_i = 1 | R_i, E] + \Pr[\neg R_i]\Pr[\bs{\rho}_i = 1 | \neg R_i, E]\\
	&= p \Pr[\bf{Z}_i | E] + (1-p)/2
	\end{align*}
where we used independence of $\bf{Z}_i$ from $R_i$
and the independence of $\bs{\rho}_i = 1$ from $E$
to obtain the final line. 
\end{proof}

\subsection{Proof of \Cref{claim:thm2_multiplicative_error}}

\densityincrease*

\begin{proof}
    We proceed by induction. When $k = 0$, $f$ is constant and so
    the claim holds trivially. Suppose $f(z) = (1 - z_i)g(z) + z_i h(z)$ where $g, h$ have depth-$k$ decision trees which don't depend on $z_i$. Note that $\E[f(\bf{U})]
    = (\E[g(\bf{U})] + \E[h(\bf{U})])/2$ and assume that $\E[g(\bf{U})]/2 \geq \E[f(\bf{U})]$. For $z \sim \bf{N}_\eps$,
    \begin{align*}
        \E[f(z)] &= \E[(1 - z_i)g(z)] + \E[z_i h(z)]\\
        &= (1 - \E[z_i])\E[g(z)] + \E[z_i]\E[h(z)] &\text{(independence)}\\
        &\leq (1/2 - \eps)(1 + \eps)^k\E[g(\bf{U})] + (1/2 + \eps)(1 + \eps)^k\E[h(\bf{U})] &\text{(induction)}\\
        &= (1 + \eps)^k\Big[(\E[g(\bf{U})] + \E[h(\bf{U})])/2 - \eps\E[g(\bf{U})] + \eps\E[h(\bf{U})]\Big]\\
        &= (1 + \eps)^k[\E[f(\bf{U})] - \eps\E[g(\bf{U})] + \eps(2\E[f{\bf{U}}] - \E[g(\bf{U})])]\\
        &= (1 + \eps)^k[\E[f(\bf{U})] + \eps(2 \E[f{\bf{U}}] - \E[g(\bf{U})])]\\
        &\leq (1 + \eps)^k[\E[f(\bf{U})] + \eps\E[f{\bf{U}}]] &\text{(assuming $\E[g]/2 \leq \E[f]$)}\\
        &= (1 + \eps)^{k+1} \cdot \E[f(\bf{U})].
    \end{align*}
    The argument for $\E[h]/2 \leq \E[f]$ is the same. 
\end{proof}

\end{document}